\newif\iflong
\newif\ifshort

\longtrue

\iflong
\else
\shorttrue
\fi

\documentclass[a4paper,UKenglish,cleveref, autoref, thm-restate]{lipics-v2021}

\usepackage{microtype}
\usepackage{graphicx}
\usepackage{tikz}
\usetikzlibrary{shapes, positioning, patterns,decorations.pathreplacing}\tikzstyle{vertex}=[circle, draw, inner sep=0pt, minimum size=4pt,outer sep = 1pt]
\newcommand{\vertex}{\node[vertex,fill]}
\usepackage{booktabs} 

\usepackage{boxedminipage}
\usepackage{cite}
\usepackage{color}
\usepackage[boxed]{algorithm2e}
\usepackage{graphicx}
\usepackage{latexsym}
\usepackage{amsmath,verbatim}
\usepackage{amssymb}
\usepackage{todonotes}

\newcommand{\cc}[1]{{\mbox{\textnormal{\textsf{#1}}}}\xspace} 
\newcommand{\NP}{\cc{NP}}
\newcommand{\FPT}{\cc{FPT}}

\newcommand{\W}[1]{\ensuremath {\cc{W}[#1]}}

\newcommand{\Nat}{\mathbb{N}}
\newcommand{\bigoh}{\mathcal{O}}
\newcommand{\XXX}{\mathcal{X}}

\title{Slim Tree-Cut Width}

\author{Robert Ganian}{Algorithms and Complexity Group, TU Wien, Vienna, Austria}{rganian@ac.tuwien.ac.at}{0000-0002-7762-8045}{Robert Ganian acknowledges support by the Austrian Science Fund (FWF, projects Y1329 and P31336).}
\author{Viktoriia Korchemna}{Algorithms and Complexity Group, TU Wien, Vienna, Austria}{vkorchemna@ac.tuwien.ac.at}{}{Viktoriia Korchemna acknowledges support by the Austrian Science Fund (FWF, project Y1329).}
\authorrunning{R.~Ganian, V.~Korchemna}
\Copyright{Anonymous Author(s)}
\hideLIPIcs

\newcommand{\width}{edge-cut width}
\newcommand{\supwidth}{super edge-cut width}
\newcommand{\stcw}{\ensuremath{\operatorname{stcw}}}
\newcommand{\gtcw}{\ensuremath{\operatorname{tcw}_0}}
\newcommand{\ecw}{\widthshort}

\newcommand{\fen}{\operatorname{fen}}
\newcommand{\widthshort}{\operatorname{ecw}}
\newcommand {\supwidthshort}{\operatorname{sec}}

\newcommand{\PP}{P}

\newcommand{\adh}{\operatorname{adh}}
\newcommand{\tcw}{\ensuremath{\operatorname{tcw}}}
\newcommand{\tor}{\operatorname{tor}}
\newcommand{\tw}{\operatorname{tw}}
\newcommand{\degtw}{\ensuremath{\operatorname{degtw}}}
\newcommand{\loc}{\operatorname{loc}}

\begin{document}
\maketitle
\begin{abstract}
Tree-cut width is a parameter that has been introduced as an attempt to obtain an analogue of treewidth for edge cuts. Unfortunately, in spite of its desirable structural properties, it turned out that tree-cut width falls short as an edge-cut based alternative to treewidth in algorithmic aspects. This has led to the very recent introduction of a simple edge-based parameter called edge-cut width [WG 2022], which has precisely the algorithmic applications one would expect from an analogue of treewidth for edge cuts, but does not have the desired structural properties.

In this paper, we study a variant of tree-cut width obtained by changing the threshold for so-called thin nodes in tree-cut decompositions from $2$ to $1$. We show that this ``slim tree-cut width'' satisfies all the requirements of an edge-cut based analogue of treewidth, both structural and algorithmic, while being less restrictive than edge-cut width. 
Our results also include an alternative characterization of slim tree-cut width via an easy-to-use spanning-tree decomposition akin to the one used for edge-cut width, a characterization of slim tree-cut width in terms of forbidden immersions as well as approximation algorithm for computing the parameter.


\end{abstract}

\section{Introduction}
Understanding which structural properties of inputs allow us to overcome the inherent intractability of problems of interest is a fundamental research area in computer science. In the context of parameterized complexity, one typically approaches this by asking which structural parameters of the input (or its graph representation) give rise to a fixed-parameter algorithm for a targeted problem. Treewidth~\cite{RobertsonSeymour86} is the most prominent example of such a structural parameter, and can be viewed as a guarantee that a graph is iteratively decomposable along small vertex separators. Many problems are known to be fixed-parameter tractable when parameterized by treewidth---and for those that are not, there is a well-studied hierarchy of more restrictive\footnote{We view parameter $\alpha$ as being more restrictive than parameter $\beta$ if every graph class where $\alpha$ is bounded also has bounded $\beta$, but the opposite does not hold.} parameters based on vertex separators or vertex deletion that can sometimes be used instead (see, e.g., Figure~1 in~\cite{BodlaenderJK13}). Examples of such parameters include the vertex cover number~\cite{FellowsLokshtanovMisraRS08,Ganian15}, the feedback vertex number~\cite{JansenB13,BergougnouxEGOR21} and treedepth~\cite{NesetrilOssonademendez12,GutinJW16,GanianO18,NederlofPSW20}. 

However, such vertex based parameters seem ill suited for handling some problems. Consider, for instance, the classical \textsc{Edge Disjoint Paths} problem (\textsc{EDP}): unlike \textsc{Vertex Disjoint Paths}, \textsc{EDP} remains \NP-hard not only on graphs of bounded treewidth, but even on graphs with a vertex cover number of at most $3$~\cite{FleszarMS18}. While this effectively rules out the use of all parameters based on vertex separators, there is an intuitive expectation that \textsc{EDP} should be fixed-parameter tractable w.r.t.\ parameters that can guarantee an iterative decomposition of the graph along small edge cuts. Indeed, \textsc{EDP} is known to be fixed-parameter tractable w.r.t.\ two basic parameterizations which provide such a guarantee: the feedback edge number~\cite{GanianO21} and treewidth plus maximum degree~\cite{GanianOR21}.

An ideal solution for handling such problems on more general inputs would be to use an alternative to treewidth that would be designed around edge cuts rather than vertex separators, one which would provide a unified justification for tractability w.r.t.\ the two basic ``edge-cut restricting'' parameterizations mentioned above. A candidate for such a parameter was proposed by Wollan, who defined \emph{tree-cut width} along with \emph{tree-cut decompositions} and described these as a variation of tree decompositions based on edge cuts instead of vertex separators~\cite{Wollan15}.
But while it is true that ``tree-cut decompositions share many of the natural properties of tree decompositions''~\cite{MarxWollan14}, from the perspective of algorithmic design tree-cut width seems to behave differently than an edge-cut based alternative to treewidth. Indeed, not only does it fall short of yielding a fixed-parameter algorithm for \textsc{EDP}~\cite{GanianO21}, it also fails to provide such algorithms for other problems one we would expect to be fixed-parameter tractable w.r.t.\ an edge-cut based analogue to treewidth. In fact, out of twelve such problems where a tree-cut width parameterization has been pursued so far, only four are fixed-parameter tractable~\cite{Ganian0S15,GanianKO21} while eight turn out to be \W{1}-hard~\cite{Ganian0S15,GozupekOPSS17,BredereckHKN19,GanianO21,GanianKorchemna21} (see the Related Work at the end of the Introduction for details). 

Very recently, Brand, Ceylan, Ganian, Hatschka and Korchemna~\cite{ECW2022} introduced a parameter called \emph{edge-cut width} which aimed at filling this gap in our understanding of edge-cut based graph parameters. On the algorithmic side, edge-cut width has precisely the properties one could hope to see in an edge-based analogue to treewidth: not only does it yield fixed-parameter algorithms for all twelve ``candidate'' problems~\cite{ECW2022}, but it is also based on a very simple type of decomposition that is much easier to use than tree-cut decompositions. That being said, already the authors of that paper noted that the structural properties of edge-cut width are far from ideal---for instance, it is the only algorithmically used parameter we are aware of that is not closed under vertex deletion. Moreover, while edge-cut width is less restrictive than the feedback edge number, unlike tree-cut width it is incomparable to treewidth plus maximum degree (even in an asymptotic sense). Because of this, it cannot act as a common generalization that would capture both of these basic approaches of enforcing decomposability along small edge cuts.

\smallskip

\noindent \textbf{Contribution.}\quad
In this paper, we identify a graph parameter which combines the advantages of tree-cut width and edge-cut width while avoiding all of the shortcomings listed above. However, before we introduce it, it will be useful to establish at least some intuitive understanding of tree-cut width\footnote{Formal definitions are provided in Section~\ref{sec:prelims}.}. 

A graph $G$ has tree-cut width $k$ if it admits a tree-cut decomposition $T$ of width $k$, whereas $T$ is a rooted tree and its nodes act as bags that form a partitioning of $V(G)$. A non-root node $t$ of $T$ defines an edge cut between all vertices in the subtree rooted at $t$, and the rest of the graph. The definition of tree-cut width then restricts, for each node $t$, the number of its children defining an edge cut of size greater than $2$. The constant ``$2$'' here arises from the structural properties Wollan aimed for when defining tree-cut width~\cite{Wollan15}; however, let us now pose the following question: How would the parameter change if we used a different constant $c$ here instead? 

On one hand, it is not difficult to observe that values of $c>2$ would immediately lead to parameters without the properties we are aiming for, since these would be constant for, e.g., all $3$-regular graphs. On the other hand, we show that for $c=0$, one obtains an asymptotically equivalent characterization of one of the previously mentioned basic edge-cut restricting parameterizations: treewidth plus maximum degree. Our parameter of interest is then the outcome of setting $c=1$; since this can be viewed as a variant of tree-cut width where all but a few children of each node need to have ``even slimmer'' edge-cuts, we refer to it as \emph{slim tree-cut width} (\stcw).

On the structural side, we show that \stcw\ inherits the desirable properties of its ``non-slim'' namesake. In particular, unlike edge-cut width~\cite{ECW2022}, \stcw\ is closed under edge sums, vertex and edge deletion, as well as under the graph immersion operation. Similarly as Wollan did for tree-cut width~\cite{Wollan15}, we also provide a set of forbidden immersions asymptotically characterizing \stcw. Furthermore, we show that \stcw\ is a common generalization of edge-cut width (and hence the feedback edge number), and treewidth plus maximum degree (see Figure~\ref{fig: hierarchy_light}).
\begin{figure}[htb]
 \begin{minipage}[c]{0.34\textwidth}
\includegraphics[width=0.9 \textwidth]{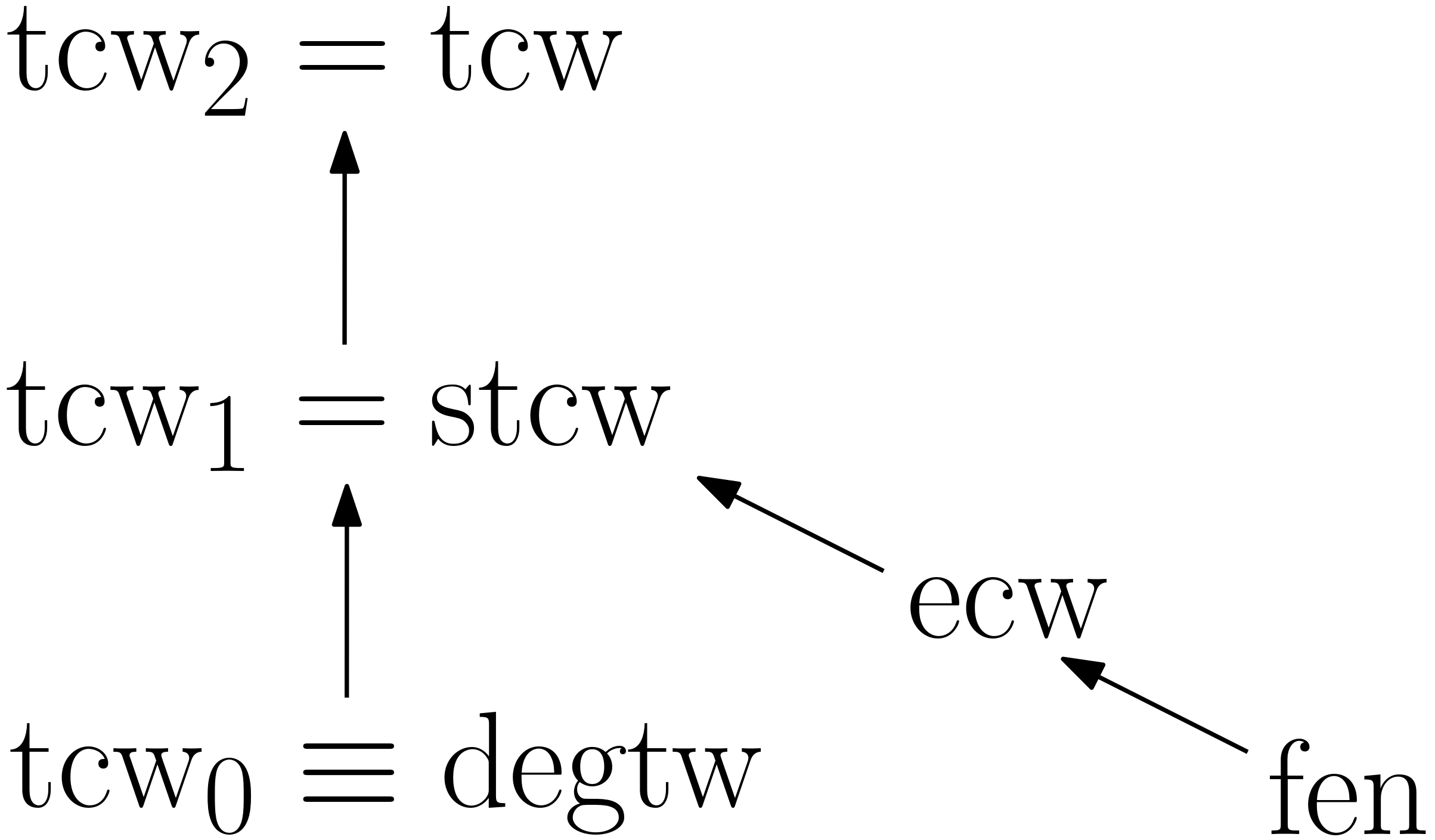}
  \end{minipage}\hfill
 \begin{minipage}[c]{0.65\textwidth}
\caption{Hierarchy of graph parameters based on edge cuts.
Here $\ecw$ denotes edge-cut width and $\degtw$ denotes treewidth plus maximum degree. $\tcw_i$ denotes the parameter obtained from tree-cut width by setting the constant $c$ described above to $i$. An arrow from $p$ to $q$ represents the fact that $p$ is more restrictive than $q$, while asymptotic equivalence is depicted by~$\equiv$.}
\label{fig: hierarchy_light}
  \end{minipage}
\end{figure}

Next, as one of our arguably most surprising results, we show that \stcw\ is asymptotically equivalent to a slight generalization of edge-cut width: instead of measuring the width over the input graph $G$, we ask for the minimum edge-cut width of any supergraph of $G$. The transformation between these parameters is constructive and has interesting algorithmic implications. First of all, when designing algorithms it allows us to avoid the use of often cumbersome tree-cut decompositions, and instead opt for the simpler decompositions used for edge-cut width---which are nothing else than spanning trees (in this case of a supergraph). Second, all of the fixed-parameter algorithms recently designed for edge-cut width~\cite{ECW2022} rely on a dynamic programming traversal of the spanning tree, and can be straightforwardly adapted to work on spanning trees of supergraphs instead. This means that one can essentially reuse the same proofs to establish fixed-parameter tractability of all considered ``candidate'' problems w.r.t.\ \stcw.

Naturally, a crucial prerequisite for algorithmically applying \stcw\ is that we can actually compute it, or more precisely compute a suitable decomposition for graphs of small \stcw.  While the problem of computing an optimal decomposition remains open even for tree-cut width, a fixed-parameter approximation algorithm was obtained by Kim, Oum, Paul, Sau and Thilikos~\cite{KimOPST18} and this suffices for the purposes of establishing fixed-parameter tractability. We obtain a similar outcome here and also provide a fixed-parameter approximation algorithm for \stcw, albeit with a worse approximation factor than for tree-cut width.


\smallskip
\noindent \textbf{Related Work.}\quad
Tree-cut width parameterizations were typically considered for problems which are not fixed-parameter tractable (\FPT) w.r.t.\ treewidth, but are \FPT\ w.r.t.\ feedback edge number and also \FPT\ w.r.t.\ treewidth plus maximum degree. The twelve candidate problems where tree-cut width parameterizations have been considered are shown in Table~\ref{tab:problems}.

\definecolor{Gray}{gray}{0.85}
\newcolumntype{a}{>{\columncolor{Gray}}c}
\newcolumntype{b}{>{\columncolor{white}}c}

\begin{table*}[htbp]
 
  \begin{center}
    \begin{tabular}{@{}l@{\quad}c@{\quad}c@{\quad}c@{\quad}c@{}
    }\toprule

Problem & tree-cut width & edge-cut width & \degtw & \stcw \\
\hline
\textsc{Capacitated Vertex Cover} & \FPT~\cite{Ganian0S15} & \FPT & \FPT& \FPT\\
\textsc{Capacitated Dominating Set} & \FPT~\cite{Ganian0S15} & \FPT & \FPT& \FPT\\
\textsc{Imbalance} & \FPT~\cite{Ganian0S15} & \FPT & \FPT& \FPT\\
\textsc{Bounded Degree Deletion} & \FPT~\cite{GanianKO21} & \FPT & \FPT& \FPT\\
\textsc{Edge Disjoint Paths} & \W{1}-hard~\cite{GanianO21} & \FPT~\cite{ECW2022} & \FPT~\cite{GanianOR21}& \FPT\\
\textsc{List Coloring} & \W{1}-hard~\cite{Ganian0S15} & \FPT~\cite{ECW2022} & \FPT~\cite{Ganian0S15}& \FPT\\
\textsc{Precoloring Extension} & \W{1}-hard~\cite{Ganian0S15} & \FPT~\cite{ECW2022} & \FPT~\cite{Ganian0S15}& \FPT\\
\textsc{Boolean Constraint Satisfaction} & \W{1}-hard~\cite{Ganian0S15} & \FPT~\cite{ECW2022} & \FPT~\cite{SamerSzeider10a}& \FPT\\
\textsc{Bayesian Network Structure Learning} & \W{1}-hard~\cite{GanianKorchemna21} & \FPT~\cite{ECW2022,GanianKorchemna21} & \FPT~\cite{OrdyniakS13}& \FPT\\
\textsc{Polytree Learning} & \W{1}-hard~\cite{GanianKorchemna21} & \FPT~\cite{ECW2022,GanianKorchemna21} & \FPT~\cite{GanianKorchemna21}& \FPT\\
\textsc{Min. Changeover Cost Arborescence} & \W{1}-hard~\cite{GozupekOPSS17} & \FPT~\cite{ECW2022} & \FPT~\cite{GozupekSSZ16}& \FPT\\
\textsc{MSRTIL}\footnotemark
 & \W{1}-hard~\cite{BredereckHKN19} & \FPT~\cite{ECW2022} & \FPT~\cite{BredereckHKN19,AdilGRSZ18}& \FPT\\

%
%
%
     \bottomrule
    \end{tabular}
  \end{center}
  \caption{\small The twelve candidate problems and their complexity w.r.t.\ edge-cut based parameters, where \degtw\ denotes the maximum degree plus treewidth. Slim tree-cut width provides a unified explanation for why these problems are \FPT\ w.r.t. both edge-cut width and \degtw, and lifts these results to more general inputs.}

  \label{tab:problems}
\end{table*}

The structural properties of tree-cut width have also been studied in a number of recent papers~\cite{GiannopoulouKRT19,GiannopoulouKRT21}. Last but not least, we note that a preprint exploring a different parameter that is aimed at providing an edge-based alternative to treewidth was recently authored by Magne, Paul, Sharma and Thilikos~\cite{edgetreewidth}; the parameter is based on different ideas and is incomparable to both tree-cut width and slim tree-cut width.

%



\section{Preliminaries}
\label{sec:prelims}

We use standard terminology for graph theory~\cite{Diestel12} and assume basic familiarity with the parameterized complexity paradigm including, in particular, the notions of \emph{fixed-parameter tractability} and \W{1}-\emph{hardness}~\cite{DowneyFellows13,CyganFKLMPPS15}. Let $\Nat$ denote the set of natural numbers including zero. We use $[i]$ to denote the set $\{0,1,\dots,i\}$. 
\footnotetext{Maximum Stable Roommates with Ties and Incomplete Lists. For completeness, we note that the authors who showed \W{1}-hardness w.r.t.\ tree-cut width also identified two additional restrictions which, when combined with tree-cut width, suffice for fixed-parameter tractability~\cite{BredereckHKN19}.}

The \emph{(open) neighborhood} of a vertex $x \in V(G)$ is the set $\{y\in V(G):xy\in E(G)\}$ and is denoted by $N_G(x)$. For a vertex subset $X$, the neighborhood of $X$ is defined as $\bigcup_{x\in X} N_G(x) \setminus X$ and denoted by $N_G(X)$; we drop the subscript if the graph is clear from the context. If $H$ is a subgraph of $G$, we denote it by $H\subseteq G$. 
\emph{Contracting} an edge $\{a,b\}$ is the operation of replacing vertices $a,b$ by a new vertex whose neighborhood is $(N(a)\cup N(b))\setminus \{a,b\}$.
For a vertex set $A$ (or edge set $B$), we use $G-A$ ($G-B$) to denote the graph obtained from
$G$ by deleting all vertices in $A$ (edges in $B$), and we use $G[A]$ to denote the
\emph{subgraph induced on} $A$, i.e., $G- (V(G)\setminus A)$. 

Let $G$ be a graph and let $x$, $y$ and $z$ be three distinct vertices of $G$ such that $(x,y),(y,z) \in E(G)$. To \emph{lift} the pair of edges $(x,y),(y,z)$ means to delete the edges $(x,y)$ and $(y,z)$ from $G$ and add (if it doesn't exist yet) a new edge $(x,z)$. We say that $G$ contains $H$ as a \emph{weak immersion} (denoted $H\le_I G$) if and only if $H$ can be obtained from $G$ by a sequence of edge deletion, vertex deletion, and lifting operations.

For a natural number $k$, we say that a graph $G$ is a \emph{$k$-edge sum} of vertex-disjoint graphs $G_1$ and $G_2$ if there exist vertices $v_i \in V(G_i)$ of degree $k$ for $i = 1,2$ and a bijection $\pi : N_{G_1}(v_1) \to N_{G_2}(v_2)$ such that $G$ is obtained from $(G_1 \setminus v_1) \cup (G_2 \setminus v_2)$ by adding an edge $(v,\pi(v))$ for every $v\in N_{G_1}(v_1)$. In this case we write $G=G_1\oplus_k G_2$. Observe that the same pair of graphs may produce different $k$-edge sums.

Given two graph parameters $\alpha,\beta: G\mapsto \Nat$, we say that $\alpha$ \emph{dominates} $\beta$ if there exists a function $p$ such that for each graph $G$, $\alpha(G)\leq p(\beta(G))$. 
If $\alpha$ dominates $\beta$ but $\beta$ does not dominate $\alpha$, we often say that $\beta$ is more restrictive than $\alpha$; as an example, treewidth dominates the vertex cover number. Two parameters that dominate each other are called asymptotically equivalent.

\smallskip
\noindent \textbf{Tree-cut Width.}\quad
The notion of tree-cut decompositions was introduced by Wollan~\cite{Wollan15}, see also subsequent work by Marx and Wollan~\cite{MarxWollan14}.
A family of subsets $X_1, \ldots, X_{k}$ of $X$ is a {\em near-partition} of $X$ if they are pairwise disjoint and $\bigcup_{i=1}^{k} X_i=X$, allowing the possibility of $X_i=\emptyset$.  

\begin{definition}
	A {\em tree-cut decomposition} of $G$ is a pair $(T,\mathcal{X})$ which consists of a rooted tree $T$ and a near-partition $\mathcal{X}=\{X_t\subseteq V(G): t\in V(T)\}$ of $V(G)$. A set in the family~$\mathcal{X}$ is called a {\em bag} of the tree-cut decomposition. 
\end{definition}

%

For any node $t$ of $T$ other than the root $r$, let $e(t)=ut$ be the unique edge incident to $t$ on the path to $r$. Let $T_u$ and $T_t$ be the two connected components in $T-e(t)$ which contain $u$ and $t$, respectively. Note that $(\bigcup_{q\in T_u} X_q, \bigcup_{q\in T_t} X_q)$ is a near-partition of $V(G)$, and we use $E_t$ to denote the set of edges with one endpoint in each part. We define the {\em adhesion} of $t$ ($\adh(t)$) as $|E_t|$;
we explicitly set $\adh(r)=0$ and $E(r)=\emptyset$. The adhesion of $(T,\mathcal{X})$ is then $\adh(T, \XXX)=\max_{t\in V(T)}\adh(t)$.

The {\em torso} of a tree-cut decomposition $(T,\mathcal{X})$ at a node $t$, written as $H_t$, is the graph obtained from $G$ as follows. If $T$ consists of a single node $t$, then the torso of $(T,\mathcal{X})$ at $t$ is $G$. Otherwise, let $T_1, \ldots , T_{\ell}$ be the connected components of $T-t$. For each $i=1,\ldots , \ell$, the vertex set $Z_i\subseteq V(G)$ is defined as the set $\bigcup_{b\in V(T_i)}X_b$. The torso $H_t$ at $t$ is obtained from $G$ by {\em consolidating} each vertex set $Z_i$ into a single vertex $z_i$ (this is also called \emph{shrinking} in the literature). Here, the operation of consolidating a vertex set~$Z$ into $z$ is to substitute $Z$ by $z$ in $G$, and for each edge $e$ between $Z$ and $v\in V(G)\setminus Z$, adding an edge $zv$ in the new graph. We note that this may create parallel edges.

The operation of {\em suppressing} (also called \emph{dissolving} in the literature) a vertex $v$ of degree at most $2$ consists of deleting~$v$, and when the degree is two, adding an edge between the neighbors of $v$. Given a connected graph $G$ and  $X\subseteq V(G)$, let the {\em 3-center} of $(G,X)$ be the unique graph obtained from $G$ by exhaustively suppressing vertices in $V(G) \setminus X$ of degree at most two. Finally, for a node $t$ of $T$, we denote by $\tilde{H}_t$ the 3-center of $(H_t,X_t)$, where $H_t$ is the torso of $(T,\mathcal{X})$ at $t$. 
Let the \emph{torso-size} $\tor(t)$ denote $|\tilde{H}_t|$. 

\begin{definition}
	The width of a tree-cut decomposition $(T,\mathcal{X})$ of $G$ is $\max_{t\in V(T)}\{ \adh(t),$ $\tor(t) \}$. The tree-cut width of $G$, or $\tcw(G)$ in short, is the minimum width of $(T,\mathcal{X})$ over all tree-cut decompositions $(T,\mathcal{X})$ of $G$.
\end{definition}

Without loss of generality, we shall assume that $X_r=\emptyset$.
We conclude this subsection with some notation related to tree-cut decompositions. 
Given a tree node $t$, let $T_t$ be the subtree of $T$ rooted at $t$. Let $Y_t=\bigcup_{b\in V(T_t)} X_b$, and let $G_t$  denote the induced subgraph $G[Y_t]$. 
A node $t\neq r$ in a rooted tree-cut decomposition is \emph{thin} if $\adh(t)\leq 2$ and \emph{bold} otherwise.


%

	A tree-cut decomposition $(T,\mathcal{X})$ is \emph{nice} if it satisfies the following condition for every thin node $t\in V(T)$: $N(Y_t)\cap (\bigcup_{b\text{ is a sibling of }t}Y_b)=\emptyset$.
	The intuition behind nice tree-cut decompositions is that we restrict the neighborhood of thin nodes in a way which facilitates dynamic programming. Every tree-cut decomposition can be transformed into a nice tree-cut decomposition of the same width in cubic time~\cite{Ganian0S15}.
	
	For a node $t$, we let $B_t=\{ b\text{ is a child of }t | |N(Y_b)|\leq 2\wedge N(Y_b)\subseteq X_t \}$ denote the set of thin children of $t$ whose neighborhood is a subset of $X_t$, and we let $A_t= \{a\text{ is a child of }t | a\not \in B_t \}$ be the set of all other children of $t$.
	Then $|A_t|\leq 2k+1$ for every node $t$ in a nice tree-cut decomposition~\cite{Ganian0S15}.

%
	
%

We refer to previous work~\cite{MarxWollan14,Wollan15,KimOPST18,Ganian0S15} for a detailed comparison of tree-cut width to other parameters. Here, we mention only that tree-cut width is dominated by treewidth and dominates treewidth plus maximum degree, which we denote $\degtw(G)$. It also dominates the feedback edge number (the size of a minimum feedback edge set), denoted $\fen(G)$.
	\begin{lemma}[\hspace{-0.001cm}\cite{Ganian0S15,MarxWollan14,Wollan15}]
		\label{lem:comparison}
		For every graph $G$, $\tw(G)\leq 2\tcw(G)^2+3\tcw(G)$ and $\tcw(G)\leq \fen(G)+1$ and $\tcw(G)\leq 4\degtw(G)^2$.
	\end{lemma}

\smallskip
\noindent \textbf{Edge-Cut Width.} \quad
The notion of edge-cut width was introduced by Brand at al.~\cite{ECW2022}.
For a graph $G$ and a maximal spanning forest $T$ of $G$, let the \emph{local feedback edge set} at $v\in V$ be\\ 
$E_{\loc}^{G,T}(v)=\{uw\in E(G)\setminus E(T)~|~$ the unique path between $u$ and $w$ in $T$ contains $v\}$.

\begin{definition}
The edge-cut width of the pair $(G,T)$ is $\widthshort(G,T)=1+\max_{v\in V} |E_{\loc}^{G,T}(v)|$, and the edge-cut width of $G$ $($denoted $\ecw(G))$ is the smallest \width\ among all possible maximal spanning forests $T$ of $G$.
\end{definition}

\begin{proposition}[\hspace{-0.001cm}\cite{ECW2022}]
\label{pro:lfescompare}
For every graph $G$, $\tcw(G)\leq \widthshort(G) \leq \fen(G)+1$.
\end{proposition}
 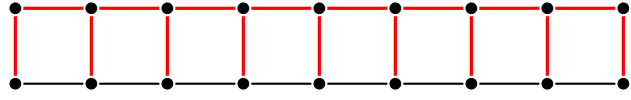
\begin{figure}
 \begin{center}
			\begin{tikzpicture}
				
				\vertex (n1) at (1,1) {};
				\vertex (n2) at (2,1) {};
				\vertex (n3) at (3,1) {};
				\vertex (n4) at (4,1) {};
				\vertex (n5) at (5,1) {};
                                           \vertex (n6) at (6,1) {};
                                           \vertex (n7) at (7,1) {};
                                           \vertex (n8) at (8,1) {};
                                           \vertex (n9) at (9,1) {}; 

                                           \vertex (m1) at (1,2) {};
				\vertex (m2) at (2,2) {};
				\vertex (m3) at (3,2) {};
				\vertex (m4) at (4,2) {};
				\vertex (m5) at (5,2) {};
                                           \vertex (m6) at (6,2) {};
                                           \vertex (m7) at (7,2) {};
                                           \vertex (m8) at (8,2) {};
                                           \vertex (m9) at (9,2) {}; 

                                           \draw [thick](n1)--(n2);
                                           \draw [thick](n2)--(n3);
                                           \draw [thick](n3)--(n4);
                                           \draw [thick](n4)--(n5);
                                           \draw [thick](n5)--(n6);
                                           \draw [thick](n6)--(n7);
                                           \draw [thick](n7)--(n8);
                                           \draw [thick](n8)--(n9);
 
                                           \draw[red, very thick] (m1)--(m2);
                                           \draw[red, very thick] (m2)--(m3);
                                           \draw[red, very thick] (m3)--(m4);
                                           \draw[red, very thick] (m4)--(m5);
                                           \draw[red, very thick] (m5)--(m6);
                                           \draw[red, very thick] (m6)--(m7);
                                           \draw[red, very thick] (m7)--(m8);
                                           \draw[red, very thick] (m8)--(m9);
                                    
				                           \draw[red, very thick] (n1)--(m1);
				\draw[red, very thick] (n2)--(m2);
                                           \draw[red, very thick] (n3)--(m3);
				\draw[red, very thick] (n4)--(m4);
                                           \draw[red, very thick] (n5)--(m5);
				\draw[red, very thick] (n6)--(m6);
                                           \draw[red, very thick] (n7)--(m7);
				\draw[red, very thick] (n8)--(m8);
                                           \draw[red, very thick] (n9)--(m9);
			\end{tikzpicture}
			\vspace{-0.5cm}
			\end{center}
\caption {Example of a graph $G$ with a spanning tree $T$ (marked in red) such that $\widthshort(G)=\widthshort(G,T)=3$. The feedback edge number of $G$ can be made arbitrarily large in this fashion.\label{fig:lfen}}
\end{figure}

Edge-cut width is not closed under vertex or edge deletions and is incomparable to $\degtw$~\cite{ECW2022}. However, the fact that its decomposition is simply a spanning tree makes it easier to work with in dynamic programming applications than, e.g., tree-cut decompositions~\cite{ECW2022}.

\section{Refined Measures for Tree-Cut Decompositions}	
\subsection{Definitions and Comparison} 
Let us now define our parameter of interest, obtained by altering the threshold for when a vertex is suppressed (dissolved) in the definition of tree-cut width. 
Formally, let $(T,\mathcal{X})$ be some tree-cut decomposition of $G$. Given a connected graph $Q$ and  $X\subseteq V(Q)$, let the {\em 2-center} of $(Q,X)$ be the unique graph obtained from $Q$ by exhaustively deleting vertices in $V(Q) \setminus X$ of degree at most one. For a node $t$ of $T$, we denote by $\bar{H}_t^2$ the 2-center of $(H_t,X_t)$, where $H_t$ is the torso of $(T,\mathcal{X})$ at $t$. Let us denote $|\bar{H}_t^2|$ by $\tor_2(t)$.

\begin{definition}
The slim width of a tree-cut decomposition $(T,\mathcal{X})$ of a graph $G$ is $\stcw(T, \XXX)=\max_{t\in V(T)}\{ \adh(t), \tor_2(t) \}$. The slim tree-cut width of $G$, or $\stcw(G)$ in short, is the minimum slim width of $(T,\mathcal{X})$ over all tree-cut decompositions $(T,\mathcal{X})$ of $G$.
\end{definition}
Observe that the difference in definitions of $\tcw(G)$ and $\stcw(G)$ is whether we dissolve the vertices of degree at most two or at most one in the torso in each node. At this point, it would be reasonable to ask what happens if we dissolve only isolated vertices (i.e., vertices of degree $0$) from the torso. Naturally extending the notions of $2$- and $3$-center for a connected graph $Q$ and  $X\subseteq V(Q)$, we define the {\em 1-center} of $(Q,X)$ as the graph obtained from $Q$ by deleting isolated vertices in $V(Q) \setminus X$. For a node $t$ of $T$, we denote by $\bar{H}_t^1$ the 1-center of $(H_t,X_t)$, where $H_t$ is the torso of $(T,\mathcal{X})$ at $t$. Let us denote $|\bar{H}_t^1|$ by $\tor_1(t)$.

\begin{definition}
The $0$-width of a tree-cut decomposition $(T,\mathcal{X})$ of $G$ is $\max_{t\in V(T)}\{ \adh(t),$ $\tor_1(t) \}$. The $0$-tree-cut width of $G$, or $\gtcw(G)$ in short, is the minimum $0$-width of $(T,\mathcal{X})$ over all tree-cut decompositions $(T,\mathcal{X})$ of $G$.
\end{definition}
It follows from the definitions that for any tree-cut decomposition $(T,\mathcal{X})$ of $G$, for each node $t$ of $T$,  $\tor(t)\le \tor_2(t) \le \tor_1(t)$. In particular, the width of $(T,\mathcal{X})$ is upper-bounded by its slim width, while the latter does not exceed the $0$-width of $(T,\mathcal{X})$. 
\begin{corollary}
\label{cor:tcw012}
For any graph $G$, $\tcw(G)\le \stcw(G) \le \gtcw(G).$
\end{corollary}
The gaps in these inequalites can be arbitrarily large---and, more strongly, $\gtcw$ is a more restrictive parameter than $\stcw$, which is in turn more restrictive than $\tcw$. Indeed, for the comparison of $\gtcw$ and $\stcw$ consider the class of stars which have slim tree-cut width $1$. Let $S_r$ denote the star with $r$ leaves (i.e., the complete bipartite graph $K_{1,r}$). 
\begin{lemma}
\label{lem: stars}
For every positive integer $r\ge 1$, $\gtcw(S_{r^2})\ge r$.
\end{lemma}
\begin{proof}
Let $(T,\XXX)$ be a tree-cut decomposition of $S_{r^2}$ of $0$-width $k$ where the bags of leaves are non-empty. Let $t$ be the node of $T$ such that $X_t$ contains the vertex of degree $r^2$. Observe that $t$ has at most $\tor_1(t)-|X_t|\le k-|X_t|$ children. For every child $t'$ of $t$, $Y_{t'}$ contains at most  $\adh(t')\le k$ vertices of $S_{r^2}$. In total, $Y_t$ contains at most $|X_t| +k \cdot ( k-|X_t|)\le k^2$ vertices of $S_{r^2}$. Together with at most $\adh(t)\le k$ verices outside of $Y_t$, $S_{r^2}$ has at most $k\cdot(k+1)$ vertices and hence $k\ge r$. 
\end{proof} 
To show the gap between $\stcw$ and $\tcw$, let us denote by $W_r$ the graph on $2r+1$ vertices consisting of $r$ triangles sharing one vertex; here we call such graphs windmills, and refer to Figure~\ref{fig:forbid_immersions} later for an illustration. The class of windmills has tree-cut width $2$ but, as the following lemma shows, unbounded slim tree-cut width.
\begin{lemma}
\label{lem: triangles}
For every positive integer $r\ge 2$, $\stcw(W_{r^2})\ge r$.
\end{lemma}
\begin{proof}
Assume, to the contrary, that there exists a tree-cut decomposition $(T,\XXX)$ of $W_{r^2}$ of slim width at most $r-1$. Let $t$ be the node of $T$ such that $X_t$ contains the vertex of degree $2r^2$. Without loss of generality, we assume that all the leaves of $T$ have non-empty bags. Then the adhesion of any child $t'$ of $t$ is at least two, as $Y_{t'}$ contains some vertex $v$ of $W_{r^2}$ and the two edge-disjoint paths from $v$ to the high-degree vertex in $t$ each contribute to $\adh(t')$.
Hence, $t$ has at most $\tor_2(t)\le r-1$ children. Moreover, for every child $t'$ of $t$, $Y_{t'}$ intersects at most $\frac{r-1}{2}$ distinct triangles of $W_{r^2}$, since each such triangle contributes $2$ to $\adh(t')$. Hence, for every child $t'$ of $t$, $Y_{t'}$ contains at most  $r-1$ vertices of $W_{r^2}$. In total, $Y_t\setminus X_t$ contains at most $(r-1)^2$ vertices of $W_{r^2}$. Since both $\adh(t)$ and $|X_t|$ are upper-bounded by $r-1$ and the former bounds the number of vertices outside of $Y_t$ by $r-1$, this would mean that $W_{r^2}$ has at most $(r-1)^2+2r-2$ vertices, a contradiction with the definition of $W_{r^2}$.
\end{proof} 
Given a graph $G$ and its nice tree-cut decomposition $(T,\XXX)$ of width at most $k$, let us denote by $B_t^{(2)}$ the set of children of $t$ from $B_t$ with adhesion precisely two; notice that $B_t^{(2)}$ does not necessarily contain all children of $t$ with adhesion precisely two, since some may lie in $A_t$. Observe that for every fixed vertex $t$ of $T$, if $x$ is an element of 2-center of the torso at $t$ and $x \not \in X_t$, then $x$ corresponds either to the parent of $t$ in $T$ or to some child of $t$ from $A_t\cup B_t^{(2)}$. Hence $\tor_2(t) \le 1+|X_t|+|A_t|+|B_t^{(2)}|\le 3k+2 + |B_t^{(2)}|$.
\begin{corollary}
\label{cor: thin2lowerbound}
Let $G$ be a graph with tree-cut decomposition $(T,\mathcal{X})$ of width at most $k$. Then for each node $t$ of $T$ it holds that
$|B_t^{(2)}|\ge \tor_2(t)-3k-2.$
\end{corollary}  

\subsection{Weak Immersions}
\label{sub: forbid_immers}
Naturally extending the result of Wollan for tree-cut width \cite{Wollan15}, we show that both slim and $0$-tree-cut width are closed under weak immersions.
\begin{theorem}
If $G$ and $H$ are graphs such that $H\le_I G$ then $\stcw(H)\le \stcw(G)$ and $\gtcw(H)\le \gtcw(G)$. 
\label{thm: monotone}
\end{theorem}
\begin{proof}
It is sufficient to proof the statement when $H$ is obtained from $G$ by precisely one edge deletion, isolated vertex deletion or lifting a pair of edges. Let $(T,\XXX)$ be a tree-cut decomposition of $G$ of minimum slim (or $0$-) width. Then $(T,\XXX)$ is also a tree-cut decomposition of $G\setminus e$ for any edge $e$ of $G$ with the same or smaller slim ($0$-) width. Similarly for the isolated vertex deletion: we just need to delete the vertex from the corresponding bag. It remains to consider the case $H=G\setminus\{(x,y),(y,z)\} \cup (x,z)$ for some $(x,y),(y,z) \in E(G)$. 

Notice that the lifting operation doesn't increase adhesion of any node $t$ of $T$: if the edge $(x,z)$ has endpoints in different connected components of $T \setminus e(t)$ then so does at least one of the edges $(x,y)$ or $(y,z)$. To see that $\tor_2(t)$ and $\tor_1(t)$ do not increase either, denote by $Q_G$ and $Q_H$ the torsos at $t$ in $(T, \XXX)$ for graphs $G$ and $H$ correspondingly. Every vertex of $Q_G$ corresponds to a non-empty subset of the vertices of
$G$. Depending on how the vertices $x$, $y$ and $z$ are split among these subsets, it holds that either $E(Q_H)\subseteq E(Q_G)$ (which yields the same or smaller 1-center and 2-center) or $Q_H$ is obtained from $Q_G$ by splitting a pair of edges. For the latter, observe that $v\in V(Q_G)\setminus X_t$ is not in the 2-center of $(Q_G, X_t)$ if and only if $v$ belongs to some induced subtree of $Q_G$ connected to the rest of $Q_G$ by at most one edge. It is not hard to see that lifting the pair of edges preserves the property. For the 1-center the situation is even simplier: isolated vertices of $Q_G$ remain isolated.
\end{proof}

Recall that the weak immersion relation $\le_I$ is a transitive, reflexive and antisymmetric relation on the set of finite graphs, i.e., a partial order. The previous theorem showed that $\stcw$ is monotone with respect to $\le_I$. Our next goal is to find graphs of simple structure but large slim (or $0$-) tree-cut width, such that forbidding them as weak immersions bounds the corresponding width of a graph. Wollan in \cite{Wollan15} characterized such graphs for tree-cut width. Namely, he established the following dichotomy:

\begin{theorem}
\label{thm: tcw2dichotomy}
(a) If $G$ is a graph such that $H_{2r^2}\le_I G$ for some $r\ge 3$, then $\tcw(G)\ge r$.
(b) There exists a function $f:N\to N$ such that if $\tcw(G)\ge f(r)$, then $H_r\le_I G$, $r\in N$.
\end{theorem}

Here $H_r$ denotes the $r$-wall, the graph which can be obtained from the $r \times r$ grid by deleting every second vertical edge in each row, see \cite{Wollan15} for the definition and Figure \ref{fig:forbid_immersions} for an illustration.  
We are going to complete the family of excluded immersions to obtain similar characterizations for $0$-tree-cut width and slim tree-cut width. 

First, we establish some forbidden weak immersions that will be useful later.
\begin{lemma}
\label{lem:forbidden} 
If $\stcw(G) < r$ then $G$ does not admit any of the following weak immersions:
\begin{itemize}
\item $r^2$ cycles intersecting at one vertex which are otherwise pairwise vertex disjoint,
\item $2r^2$ paths with the same endpoints which are otherwise pairwise vertex disjoint.
\end{itemize}
\end{lemma}
\begin{proof}
We will show that whenever $G$ admits any of the listed immersions, it also admits an immersion of $W_{r^2}$. The statement will then follow from Lemma \ref{lem: triangles} and Theorem \ref{thm: monotone}. The case of $2r^2$ paths can be reduced to the case of $r^2$ cycles as follows: let $x$ and $y$ be the endpoints, we arbitrarily form $r^2$ pairs of paths and in each pair lift the edges adjacent to $y$. The resulting graph consists of $r^2$ cycles intersecting at $x$. Further, we consequently lift the pairs of edges of each cycle to make it a triangle, which results in $W_{r^2}$. 
\end{proof}

\begin{figure}[htb]
\includegraphics[width=\textwidth]{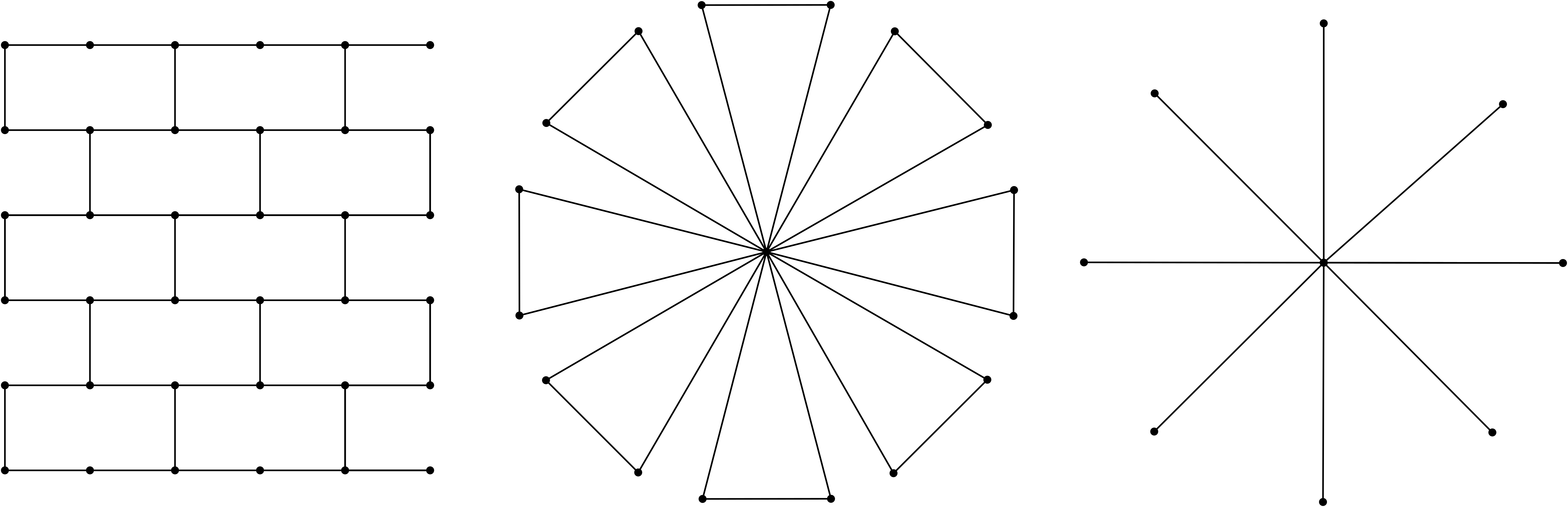}
\caption{Illustrations of forbidden weak immersions for the graphs with bounded standard, slim or $0$-tree-cut width. Left: 6-wall $H_6$, Middle: windmill $W_8$, Right: star $S_8$.} 
\label{fig:forbid_immersions}
\end{figure}

We already know that the families of stars $S_r$ and windmils $W_r$ have unbounded $0$- and slim tree-cut width, respectively. As we will show in the remainder of this subsection, excluding $W_r$ ($S_r$) as weak immersions along with $H_r$ guarantees bounded slim tree-cut width ($0$-tree-cut width).  

\begin{theorem}
Let $G$ be a graph and $r \ge 1$ a positive integer. If $H_{2r^2}\le_I G$ or $S_{r^2}\le_I G$ for some $r\ge 3$, then $\gtcw(G)\ge r$.
Moreover, there exists a function $h:N\to N$ such that if $\gtcw(G)\ge h(r)$, then $H_r\le_I G$ or $S_{r}\le_I G$.
\end{theorem}
\begin{proof}
If $H_{2r^2}\le_I G$ for some $r\ge 3$, we have that $\tcw(G)\ge r$ by Theorem \ref{thm: tcw2dichotomy} and hence $\gtcw(G)\ge r$. In case $S_{r^2}\le_I G$, the lower bound follows from Theorem \ref{thm: monotone} and Lemma \ref{lem: stars}.

Let $f$ be the function given by Theorem \ref{thm: tcw2dichotomy}. We define $h$ by setting $h(r)=r\cdot f(r)+3\cdot f(r)+2$. Assume that $G$ is a graph such that $\gtcw(G) \ge h(r)$. If $\tcw(G) \ge f(r)$, we immediatedly conclude that $H_r\le_I G$ by Theorem \ref{thm: tcw2dichotomy}. Otherwise, let $(T,\XXX)$ be a nice tree-cut decomposition of $G$ of width at most $f(r)$ with leaves having non-empty bags. There exists a node $t$ of $T$ such that $\tor_1(t)\ge h(r)$, in particular, $B_t\ge r \cdot f(r)$. As the size of $X_t$ is at most $f(r)$, some vertex of $X_t$ has degree of at least $r$ and hence $S_{r}\le_I G$. 
\end{proof}

Before providing similar characterization for slim tree-cut width, we introduce a simple technical modification of tree-cut decompositions, which will also be used later for 
establishing the connection between slim tree-cut width and edge-cut width. The aim is, roughly speaking, to avoid the situation where a thin child has adhesion $2$, even though it consists of two completely independent components each of which could be a thin child of adhesion $1$.
Formally, let $(T, \XXX)$ be a nice tree-cut decomposition of $G$. We say that a node $t$ with parent $t'$ in $T$ is \emph{decomposable} if the following conditions hold: 
\begin{itemize}
 \item $t\in B_{t'}$ and there exist two edges $e_1$ and $e_2$ between $G_t$ and $G\setminus G_t$ in $G$;
 \item the endpoints of $e_1$ and $e_2$ in $G_t$ belong to different connected components of $G_t$.
\end{itemize}
\begin{lemma}
\label{lem: no_decomposable}Any nice tree-cut decomposition of $G$ can be transformed into a nice tree-cut decomposition of the same tree-cut width with no decomposable nodes. 
\end{lemma}
\begin{proof}
Let $(T', \XXX')$ be a nice tree-cut decomposition of $G$ with at least one decomposable node. Let $t$ be a decomposable node of $T'$ with minimum distance to the root, and let $e_1$ and $e_2$ be the edges between $G_t$ and $G\setminus G_t$ in $G$. We create a copy $T'_{t^*}$ of the rooted subtree $T'_t$ where the copy of $s\in T'_t$ is $s^*\in T'_{t^*}$. We then connect $t^*$ to the parent of $t$. Let $G_1$ be the connected component of $G_t$ containing an endpoint of $e_1$. For every $s \in V(T'_t)$ we set $X_s=X'_s \cap V(G_1)$ and $X_{s^*}=X'_s\setminus X_s$. For the rest of nodes $s$ of $T'$ we set $X_s=X'_s$. Finally, we exhaustively remove empty bags which are leaves and denote the obtained tree by $T$. Observe that the resulting decomposition $(T,\XXX)$ is nice and its width is not greater than the width of $(T', \XXX')$. Moreover, our transformation doesn't create any decomposable nodes outside of subtrees rooted in $t$ and $t^*$; both $t$ and $t^*$ have an adhesion of one and hence are not decomposable. Therefore, after a finite number of such steps we obtain some nice tree-cut decomposition of $G$ of the same width but with no decomposable nodes.
\end{proof}
Further, as a technical term, we will refer to nice decompositions with no decomposable nodes as \emph{very nice} decompositions. 
\begin{corollary}
\label{cor: verynice}
Every tree-cut decomposition can be transformed into a very nice tree-cut decomposition in quartic time, without increasing the width. 
\end{corollary}
\begin{proof}
Let $(T'',\mathcal{X}'')$ be a tree-cut decomposition of $G$ of width $k$. We transform $(T'',\mathcal{X}'')$ into a nice tree-cut decomposition $(T',\mathcal{X}')$ of width at most $k$ (this can be done in cubic time, see \cite{Ganian0S15} for details). Further, we apply Lemma \ref{lem: no_decomposable} on $(T',\mathcal{X}')$. This requires at most quartic time, since every node of $T'$ is decomposed at most once and every such decomposition can be performed in cubic time. Then the resulting decomposition $(T,\mathcal{X})$ is very nice and has width of at most $k$.
\end{proof}
With this transformation in hand, we are now ready to fully characterize forbidden weak immersions for graphs of bounded slim tree-cut width.
\begin{theorem}
Let $G$ be a graph and $r \ge 1$ a positive integer. If $H_{2r^2}\le_I G$ or $W_{r^2}\le_I G$ for some $r\ge 3$, then $\stcw(G)\ge r$.
Moreover, there exists a function $g:N\to N$ such that if $\stcw(G)\ge g(r)$, then $H_r\le_I G$ or $W_{r}\le_I G$.
\end{theorem}
\begin{proof}
If $H_{2r^2}\le_I G$ for some $r\ge 3$, we have that $\tcw_2(G)\ge r$ by Theorem \ref{thm: tcw2dichotomy} and hence $\stcw(G)\ge r$. In case $W_{r^2}\le_I G$, the lower bound follows from Lemma~\ref{lem:forbidden}.

Let $f$ be the function given by Theorem \ref{thm: tcw2dichotomy}. We define $g$ by setting $g(r)=2r\cdot f^2(r)+3\cdot f(r)+2$. Assume that $G$ is a graph such that $\stcw(G) \ge g(r)$. If $\tcw(G) \ge f(r)$, we immediatedly conclude that $H_r\le_I G$ by Theorem \ref{thm: tcw2dichotomy}. Otherwise, by Corollary \ref{cor: verynice} there exists a very nice tree-cut decomposition $(T,\XXX)$ of $G$ of width at most $f(r)$. Let us pick a node $t$ of $T$ such that $\tor_2(t)\ge g(r)$. By Corollary \ref {cor: thin2lowerbound} we have that $|B_t^{(2)}| \ge g(r)-3 \cdot f(r)-2 = 2r\cdot f^2(r)$. Since $(T,\XXX)$ is very nice, all the children of $t$ in $B_t^{(2)}$ are non-decomposable. Recall that for every $t'\in B_t^{(2)}$, the neighbourhood of $Y_{t'}$ in $G$ is a one- or two-element subset of $X_t$, and hence $Y_{t'}$ provides a path between some (possibly equal) vertices of $X_t$. As the size of $X_t$ is at most $f(r)$, $G$ contains either $2r$ cycles intersecting in one vertex of $X_t$ or $2r$ paths between two vertices of $X_t$. Analogously to the proof of Lemma~\ref{lem:forbidden}, in both cases $W_{r}\le_I G$.
\end{proof}
\subsection{$k$-Edge Sums} 
Another natural property Wollan~\cite{Wollan15} established for tree-cut width is that the parameter is closed under the operation of taking $k$-edge sum for small $k$. Specifically, he proved the following:
\begin{lemma}[\hspace{-0.0001cm}\cite{Wollan15}]
Let $G$, $G_1$, and $G_2$ be graphs such that $G = G_1\oplus_k G_2$. If $G_j$ has a tree-cut
decomposition $(T_j, \XXX_j)$ for $j = 1,2$, then $G$ has a tree-cut decomposition $(T,\XXX)$ such that $\adh(T,\XXX)=\max\{k,\adh(T_1,\XXX_1), \adh(T_2, \XXX_2)\}$. Moreover, for every $t\in V(T)$, the torso $H_t$ of $t$ in $(T,\XXX)$ is isomorphic to the torso of some vertex of $(T_1,\XXX_1)$ or $(T_2,\XXX_2)$. 
\end{lemma}
Based on this result for optimal decompositions $(T_1,\XXX_1)$ and $(T_2,\XXX_2)$, we immediatedly obtain the upper bound on 0- and slim tree-cut width for $k$-edge sums:
\begin{corollary}
Let $G$, $G_1$ and $G_2$ be graphs such that $G = G_1\oplus_k G_2$. Then it holds that $\stcw(G) \le \max\{k, \stcw(G_1), \stcw(G_2)\}$ and $\gtcw(G) \le \max\{k,\gtcw(G_1), \gtcw(G_2)\}$.
\end{corollary}
In particular, if both $G_1$ and $G_2$ have $0$-, slim or standard width of at most $\omega$ and $k \le \omega$, we may conclude that the corresponding width of $G$ is at most $\omega$. 

\section{Alternative Characterizations}
In this section, we study alternative characterizations of slim tree-cut width and $0$-tree-cut width. In particular, we observe that the latter is asymptotically equivalent to maximum degree plus treewidth. This provides an interesting connection between tree decompositions and tree-cut decompositions, but essentially rules out its study as a means of establishing novel tractability results. For slim tree-cut width, however, we obtain a characterization that ties it to the previously studied edge-cut width and has algorithmic implications.

\subsection{Characterization of 0-Tree-Cut Width}
Wollan~\cite{Wollan15} showed that a bound on the treewidth and maximum degree implies a bound on the tree-cut width of a graph:
\begin{proposition}
Let $G$ be a graph with maximal degree $d$ and
treewidth $w$. Then there exists a tree-cut decomposition of adhesion at most
$(2w+2)d$ such that every torso has at most $(d+1)(w+1)$ vertices. 
\end{proposition}
In particular, as  $\tor_1(t)\le|H_t|\le (d+1)(w+1) \le (2w+2)d$ for every node $t$ of $T$, we have $\gtcw(G)\le (2w + 2)d$. In the following proposition, we show that the converse is true as well: bounded $\gtcw$ implies bounded treewidth and maximum degree of a graph.

\begin{proposition}
Let $G$ be a graph with $\gtcw(G) = k$. Then every vertex of $G$ has degree of at most $k^2$ and $\tw(G)\leq 2k^2+3k$.
\end{proposition}
\begin{proof}
By Lemma \ref{lem:comparison} and Corollary \ref{cor:tcw012} we have that $\tw(G)\leq 2 \tcw(G)^2+3 \tcw(G) \leq 2k^2+3k$.
Since $\gtcw(G) = k$, Lemma \ref{lem: stars} implies that $G$ does not contain $S_{(k+1)^2}$ as a weak immersion, in particular, degree of any vertex of $G$ is at most $k^2$.  
\end{proof}

\begin{corollary}
$0$-tree-cut width is asymptotically equivalent to maximum degree plus treewidth. 
\end{corollary}

\subsection{Characterization of Slim Tree-Cut Width}

Recall that edge-cut width is a parameter that is defined over spanning trees in the input graph $G$, which serve as the corresponding decompositions. Let us now consider a slight generalization of this where we consider not only spanning trees over $G$, but of any supergraph of $G$. Such a generalization would---unlike edge-cut width itself---trivially be closed under both vertex and edge deletion.
For our considerations, let us denote this parameter \emph{\supwidth} ($\supwidthshort(G)$): $$\supwidthshort(G)=\min\{\widthshort(H, T)| H\supseteq G \text{ and } T\text{ is a spanning forest of }H\}.$$ If $H\supseteq G$ is a supergraph of $G$ and $T$ is a spanning forest of $H$ such that $\widthshort(H, T) \le k$, we say that $T$ witnesses $\supwidthshort(G) \le k$. Observe that there always exists a connected witness, i.e., a tree. Indeed, if $H$ consists of $m>1$ connected components, we can arbitrarily extend it to a connected graph $H^*$ by adding $m-1$ edges. The addition of these edges to $T$ then results in the tree $T^*$ witnessing $\supwidthshort(G)\le k$. Moreover, notice that any witness of $\widthshort(G)\le k$ is also a witness of $\supwidthshort(G)\le k$.

\begin{corollary}
\label{cor:sec_ecw}
For every graph $G$, $\supwidthshort(G)\leq \widthshort(G)$.
\end{corollary}
A slight modification of the proof of Proposition \ref{pro:lfescompare} yields the following lower bound:
\begin{proposition}
\label{pro:tcw_sec}
For every graph $G$, $\tcw(G)\leq \supwidthshort(G)$.
\end{proposition}
\begin{proof}
Let $Q$ be the supergraph of $G$ and let $T$ be the spanning tree of $Q$ such that $\widthshort(Q,T)=\supwidthshort(G)$. We construct a tree-cut decomposition $(T, \XXX)$ of $G$ where each bag contains at most one vertex, notably by setting $X_t=\{t\}$ for each $t\in V(G)$ and $X_t=\emptyset$ for each $t\in V(Q)\setminus V(G)$.  Fix any node $t$ in $T$ other than the root, let $u$ be the parent of $t$ in $T$.  All the edges of $G\setminus ut$ with one endpoint in the rooted subtree $T_t$ and another outside of $T_t$ belong to $E^{Q,T}_{loc}(t)$, so $\adh_T(t)\leq |E^{Q,T}_{loc}(t)|+1 \leq \supwidthshort(G)$.

Let $H_t$ be the torso of $(T, \XXX)$ in $t$, then $V(H_t)=X_t \cup \{z_1...z_l\}$ where $z_i$ correspond to connected components of $T \setminus t$, $i\in [l]$. In $\tilde H_t$, only $z_i$ with degree at least $3$ are preserved. 
But all such $z_i$ are the endpoints of at least two edges in $|E^{Q,T}_{loc}(t)|$, so $\tor(t)=|V(\tilde H_t)|\le 1+ |E^{Q,T}_{loc}(t)| \le \supwidthshort(G)$. Thus  $\tcw(G)\leq \supwidthshort(G)$. 
\end{proof}
To represent a deeper connection between tree-cut decompositions and super edge-cut width, it will be convenient to work with very nice decompositions introduced in subsection \ref{sub: forbid_immers}. 
\begin{proposition}
\label{prop: eq_tree-cutwidth1}Let $(T, \XXX)$ be a very nice tree-cut decomposition of $G$ of width at most $k$. Then for each node $t$ of $T$, $|B_t^{(2)}| \le k \cdot \supwidthshort(G)$. In particular, $\stcw(G) \le \supwidthshort(G)^2+4 \cdot \supwidthshort(G)$.
\end{proposition}

\begin{proof}
Assume that $T^*$ is a spanning tree of $H\supseteq G$ such that $\supwidthshort(G)=\widthshort(H, T^*)$.
 For any node $t$ of $T$ and $b\in B_t^{(2)}$, $b$ has one of three types (see Figure \ref{fig: tcw_cases}):
\begin{enumerate}
\item $N(Y_b)=\{x\}$ for some $x \in X_t$, $x$ is connected to distinct $x_b^1$ and $x_b^2$ from $Y_b$;
\item $N(Y_b)=\{x_1,x_2\}$ for $x_1\ne x_2$, $x_1$ and $x_2$ are connected to the same $x_b\in Y_b$;
\item $N(Y_b)=\{x_1,x_2\}$ for $x_1\ne x_2$, $x_1$ and $x_2$ are connected to distinct $x_b^1$ and $x_b^2$ from $Y_b$ correspondingly;
\end{enumerate}
\begin{figure}[htb]
\begin{center}
\includegraphics[width=0.7 \textwidth]{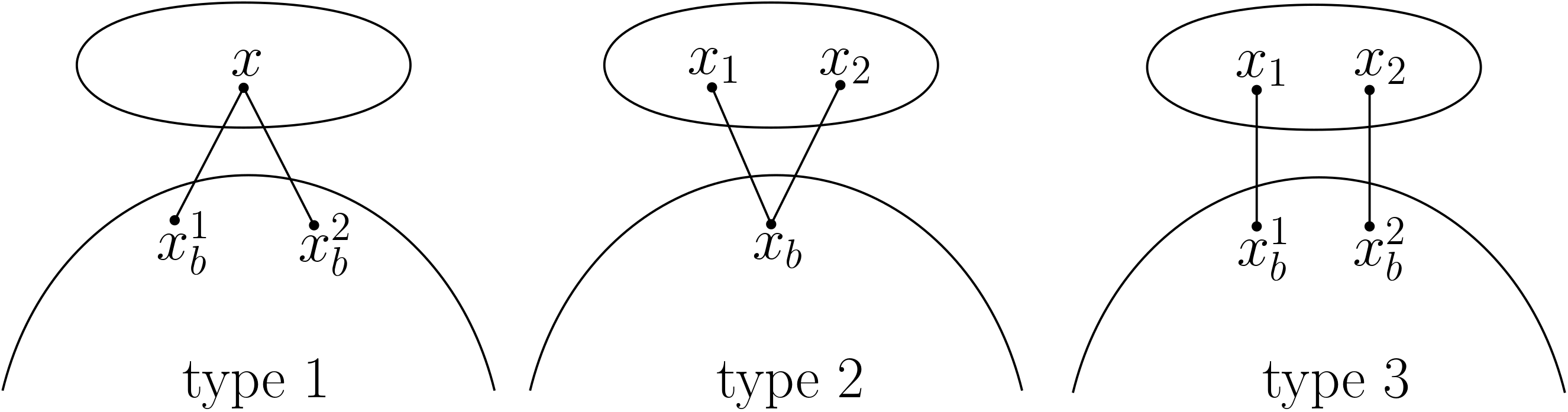}
\end{center}
\caption{Possible configurations of edges between thin child $b\in B_t^{(2)}$ and its parent $t$.}
\label{fig: tcw_cases}
\end{figure}
Let us start with the first type. If $x_b^ix$ doesn't belong to $T^*$ for $i=1$ or $i=2$, then $x_b^ix \in E^{H,T^*}_{loc}(x)$. Otherwise, $x_b^1$ and $x_b^2$ are connected via $x$ in $T^*$. Then $T^*[Y_b]$ has precisely two connected components. As $b$ is not decomposable, there exists a path $p$ between $x_b^1$ and $x_b^2$ in $G_b$ containing precisely one edge outside of $T^*$. This edge contributes  to $E^{H,T^*}_{loc}(x)$.  

As $T^*$ is a tree, there can be at most $|X_t|-1\le k-1$ thin children $b$ of the second type such that $x_b$ is adjacent to two elements of $X_t$ in $T^*$. For the remaining vertices $b$ of the second type, there exists $x\in X_t$ such that $xx_b\in G\setminus T^* \subseteq H\setminus T^*$ and therefore $xx_b \in E^{H,T^*}_{loc}(x)$. 

Let $b$ be a thin node of the third type. If $x^b_1$ and $x^b_2$ are connected via a path in $T^*[Y_b]$, we can apply the same argument as for the second type. Otherwise, $T^*[Y_b]$ has precisely two connected components and, analogously to the first type, there exists an edge in $G_b\cup\{x_1 x^b_1, x_2 x^b_2\}$ that belongs to $E^{H,T^*}_{loc}(x_1)$.  

To conclude, any node of $B_t^{(2)}$ either increases $E^{H,T^*}_{loc}(x)$ for some $x\in X_t$ or creates a path in $T^*$ between two vertices of $X_t$. Since $T^*$ is a tree, $|X_t|\le k$ and $|E^{H,T^*}_{loc}(x)|\le \supwidthshort(G)-1$ for every $x\in X_t$, the size of $B_t^{(2)}$ is at most $(k-1)+\sum_{x\in X_t}|E^{H,T^*}_{loc}(x)|\le k\cdot \supwidthshort(G)-1$. Then $\tor_2(t) \le |A_t|+|X_t|+1+|B_t^{(2)}| \le 3k+1+k\cdot \supwidthshort(G) \le k\cdot (\supwidthshort(G)+4)$. Since the bound holds for every node $t$ of $T$, we may conclude that the slim width of $(T, \XXX)$ is at most $k\cdot (\supwidthshort(G)+4)$. By Proposition \ref{pro:tcw_sec} and Corollary \ref{cor: verynice}, there exists a very nice tree-cut decomposition of $G$ of width $k\le \supwidthshort(G)$, therefore $\stcw(G) \le \supwidthshort(G)^2+4\cdot \supwidthshort(G)$.  
\end{proof}
Hence, slim tree-cut with of any graph is upper-bounded by a quadratic function of its super edge-cut width. Next, we show that the converse statement holds as well:
\begin{proposition}
\label{prop: eq_tree-cutwidth2}
For every graph $G$, $\supwidthshort(G)\le 3\cdot (\stcw(G)+1)^2$. Moreover, given a tree-cut decomposition of $G$ of slim width $k$, it is possible to compute a supergraph $Q\supseteq G$ and its spanning tree $T$ witnessing $\supwidthshort(G)\leq 3(k+1)^2$ in cubic time. 
\end{proposition}
\begin{proof}
Let $(T_0,\mathcal{X}_0)$ be a tree-cut decomposition of $G$ of slim width $k$. We start by transforming it into a nice tree-cut decomposition $(T,\mathcal{X})$ in cubic time as in \cite{Ganian0S15}. The transformation procedure acts on the 2-centers of torsos only by contracting some edges. Recall that $v\in V(H_t)\setminus X_t$ is not in the 2-center of $(H_t, X_t)$ if and only if $v$ belongs to some induced subtree of $H_t$ connected to the rest of $H_t$ by at most one edge. Since contracting an edge either preserves the property or merges $v$ with some other vertex, it doesn't increase $\tor_2(t)$ for any node $t$ of $T$. In particular, the slim width of $(T,\mathcal{X})$ is at most $k$. 

Let $\Omega\subseteq \mathcal{X}$ be the set of empty bags of $(T,\mathcal{X})$, we construct $Q\supseteq G$ along with its tree-cut decomposition $(T,\mathcal{X'})$ as follows. Firstly, we add to $G$ vertices $v_t$ for every $t\in \Omega$. We define $X'_t=\{v_t\}$ if $X_t=\emptyset$ and $X'_t=X_t$ otherwise. For every node $t\in T$, construct an arbitrary tree $T^*_t$ over $X'_t$ and add its edges to $Q$. Further, we process every edge $e=pt \in E(T)$ such that $p$ is the parent of $t$ in $T$ and either $N(Y_t)\not \subseteq X_t$ or $\adh(t)>1$ as follows. If $G$ doesn't contain an edge between $X_t'$ and $X'_p$, we add to $E(Q)$ arbitrary edge with endpoints in $X_t'$ and $X'_p$. This increases the adhesion of $e$ by at most one. 

Now we proceed to the choice of the spanning tree $T^*$ in $Q$. 
For every $t\in T$ other then the root, let $p$ be the parent of $t$ in $T$. If $\adh(t)=1$ and  $N(Y_t) \subseteq X_t$, we denote by $e_t$ the unique edge between $Y'_t$ and $X'_p$ in $Q$. Otherwise, let $e_t$ be arbitrary edge of $Q$ with endpoints  in $X_t'$ and $X'_p$. We then construct $T^*$ by gluing together all $T^*_t$ via edges $e_t$: $T^*=(\cup_{t\in V(T)}T^*_t)\bigcup (\cup_{t\in V(T)\setminus r} \{e_t\})$. Obviously the construction can be performed in cubic time; we will show that $\supwidthshort(Q, T^*) \le 3(k+1)^2$. 

To this end, fix any node $t$ of $T$ and $x\in X'_t$ and denote $E_{loc}(x)=E_{loc}^{Q,T^*}(x)$. If $T^*$ contains more than one edge between $Y'_t$ and rest of  $T^*$, then all but one of them are the unique edges connecting $Q'_q$ to the rest of $Q$ for some descendants $q$ of $t$ in $T$. Hence, they don't belong to any path in $T^*$ between the endpoints of some feedback edge $e\in E(Q)\setminus E(T^*)$. Therefore, every edge of $ E_{loc}(x)$ has at least one endpoint in $Y'_t$. The number of edges in $E_{loc}(x)$ with both endpoints in $X_t'$ is at most $|X_t'|\cdot(|X_t'|-1)\le k\cdot(k-1)$. Every edge with one endpoint in $X_t'$ and another outside of $Y_t'$ contributes to the adhesion of $t$ in $(T,\mathcal{X'})$, so their number is bounded by $k+1$.

Finally, if $e=yz \in E_{loc}(x)$ contains an endpoint $y$ in $Y'_t \setminus X'_t$, then $y\in Y'_q$ for some child $q$ of $t$. Then $Q$ contains a cycle intersecting $Y'_q$ and $x\in X_t$. In particular, by construction of $Q$ we may conclude $q\in A_t \cup B_t^{(2)}$ w.r.t. the decomposition $(T,\mathcal{X})$. By the same arguments as for the node $t$, we conclude that at most one edge between $Y'_q$ and the rest of $T^*$ belongs to any path in $T^*$ between the endpoints of some feedback edge $e\in E(Q)\setminus E(T^*)$, so $z\not \in Y'_q$ and $e$ contributes to the adhesion of $q$ in $(T,\mathcal{X'})$. In particular, $E_{loc}(x)$ contains at most $\adh(q)+1$ edges with an endpoint in $Y'_q$. In total, at most $\max_{q\in A_t}(\adh(q)+1)\cdot|A_t| +\max_{q\in B_t^{(2)}}(\adh(q)+1)\cdot |B_t^{(2)}| \le (k+1)(2k+1)+3k=2k^2+6k+1$ edges in $E_{loc}(x)$ have an endpoint in $Y'_{t}\setminus X'_t$, so $|E_{loc}(x)|\le k\cdot(k-1)+(k+1)+2k^2+6k+1 =3k^2+6k+2$ and hence $\sec(Q, T^*)\le 3k^2+6k+3 =3(k+1)^2$.
\end{proof}

\begin{corollary}
\label{cor:equiv}
$\supwidthshort$ and $\stcw$ are asymptotically equivalent.
\end{corollary}

The results of this section are summarized in Figure~\ref{fig: hierarchy_ext}.

\begin{figure}[htb]
 \begin{minipage}[c]{0.23\textwidth}
\includegraphics[width=\textwidth]{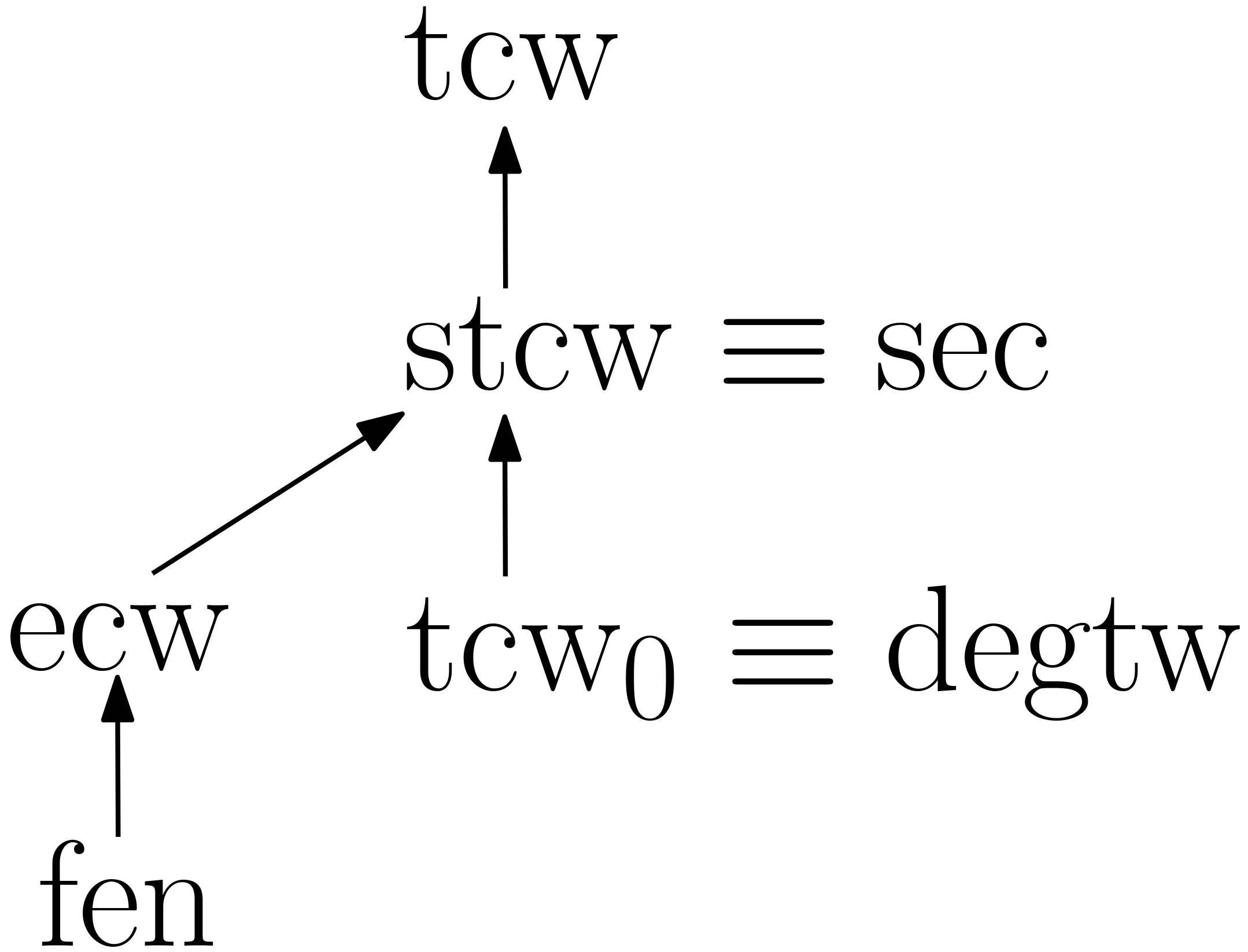}
  \end{minipage}\hfill
 \begin{minipage}[c]{0.7\textwidth}
\caption{Position of slim and 0-tree-cut width in the
hierarchy of edge-cut based parameters. 
An arrow from $p$ to $q$ represents the fact that $p$ is more restrictive than $q$, while asymptotic equivalence is depicted by~$\equiv$.} 
\label{fig: hierarchy_ext}
  \end{minipage}
\end{figure}

\section{Approximating Slim Tree-Cut Width}

In this section we show how to efficiently construct a tree-cut decomposition of a graph $G$ with slim width bounded by a cubic function of its optimal value $\stcw(G)$. As a starting point for our approximation, we use the following result of Kim, Oum, Paul, Sau and Thilikos:
\begin{theorem}[\hspace{-0.0001cm}\cite{KimOPST18}]
\label{thm: tcw_approx}
There exists an algorithm that, given a graph $G$ and $\omega \in \Nat$,
either outputs a tree-cut decomposition of $G$ with width at most $2\omega$ or correctly reports that no tree-cut decomposition of $G$ with width at most $\omega$ exists in $2^{\bigoh(\omega^2 \cdot log \omega)}\cdot n^2$ steps.
\end{theorem}
As an observant reader might have already noticed, if $G$ has bounded slim tree-cut width, it imposes some restrictions on the structure of possible decompositions of $G$ of small (standard) tree-cut width. This fact enables us to construct an efficient approximation  for $\stcw(G)$.

\begin{theorem}
There exists an algorithm that, given a graph $G$ and $\omega \in \Nat$,
either outputs a tree-cut decomposition of $G$ with slim width at most $6(\omega+1)^3$ or correctly reports that no tree-cut decomposition of $G$ with slim width at most $\omega$ exists in $2^{\bigoh(\omega^2 \cdot log \omega)}\cdot n^4$ steps.
\end{theorem}

\begin{proof}
Given a graph $G$ and $\omega \in \Nat$, let us run the algorithm from Theorem \ref{thm: tcw_approx}. If it reports that $\tcw(G)>\omega$, we may conclude that $\stcw(G)>\omega$ by Corollary \ref{cor:tcw012}. In case the algorithm returns a tree-cut decomposition $(T', \XXX')$ of width at most $2 \omega$, we invoke Corollary \ref{cor: verynice} to transform this decomposition into a very nice decomposition $(T, \XXX)$ of the same width in at most quartic time. By Proposition \ref{prop: eq_tree-cutwidth1}, we have that $|B_t^{(2)}|\le 2\omega \cdot \supwidthshort(G)$ for each node $t$ of $T$. If for some node $t$ the size of $B_t^{(2)}$ exceeds $6\omega \cdot (\omega+1)^2$, then $\supwidthshort(G)>3(\omega+1)^2$ and by Proposition \ref{prop: eq_tree-cutwidth2} we may correctly report that $\stcw(G)>\omega$. Otherwise, $\tor_2(t)\le 1+|X_t|+|A_t|+|B_t^{(2)}|\le 1+2\omega +(4 \omega + 1)+6\omega \cdot (\omega+1)^2 \le 6(\omega+1)^3$ for any node $t$ of $T$. Hence, the slim width of $(T, \XXX)$ is at most $6(\omega+1)^3$.
\end{proof}

\section{Discussion of Algorithmic Applications}
Having established its structural properties, we now turn to the algorithmic aspects of slim tree-cut width. Here, Corollary~\ref{cor:equiv} shows that instead of using a tree-cut decomposition of the input graph $G$ to design fixed-parameter algorithms---as was done in past dynamic programming algorithms that utilized tree-cut width---we can perform dynamic programming along a spanning tree $T$ of a supergraph $Q$ of $G$. Both $Q$ and $T$ can be computed from $G$ in a pre-processing stage by using Proposition~\ref{prop: eq_tree-cutwidth2}, and using a spanning tree instead of a tree-cut decomposition typically leads to significantly more concise (and conceptually cleaner) algorithms. 

The cost for this simplification is the quadratic gap between the widths of these decompositions. We note that this situation is somewhat analogous to how one still typically uses clique-width~\cite{CourcelleMakowskyRotics00} as a general and easy-to-use parameterization for various problems (especially when aiming for instances with higher edge-densities), even though rank-width~\cite{Oum05} and Boolean-width~\cite{BuiXuanTV11} are asymptotically equivalent parameterizations which have been shown to yield more efficient algorithms~\cite{GanianHlineny10}---there, the gap is even exponential.

Recall that 
a number of problems which remain \W{1}-hard w.r.t.\ tree-cut width have recently been shown to be fixed-parameter tractable when parameterized by edge-cut width~\cite{ECW2022,GanianKorchemna21}, via explicit dynamic programming algorithms which proceed along the spanning tree of the input graph. While the functional gap between edge-cut width and \supwidth\ (and, analogously, slim tree-cut width) may be arbitrarily large, it is not difficult to see that each of the algorithms provided in those papers can be straightforwardly lifted to fixed-parameter algorithms w.r.t.\ \supwidth. Indeed, the only amendment one needs to make is to deal with the presence of ``ghost'' edges and vertices which occur in the spanning tree but not in the graph, and the computation of the records in these algorithms can easily deal with such vertices and edges. 

To provide a concrete illustration of how this can be done, let us revisit the dynamic programming algorithm for the \textsc{Edge Disjoint Paths} problem parameterized by edge-cut width~\cite[Theorem 2]{ECW2022}. No change is needed to the records. When the algorithm attempts to compute the set of ``valid records'' for a vertex $v$ from the sets of valid records for some of its children $v_1,\dots,v_\psi$ in the spanning tree, the algorithm performs a branching step in which it considers all possible ways the paths can be routed between the subtrees rooted at these children (See the ``If v is an internal node'' paragraph in the proof). At this branching step, we simply discard all routings which use edges that are not present in $G$. The situation is no more complicated for the other considered problems---in essentially all cases, the change simply boils down to ignoring the vertices and edges which do not exist in $G$.

Hence, we obtain:

\begin{corollary}[Theorems 2-6 in~\cite{ECW2022}, Theorems 6 and 14 in~\cite{GanianKorchemna21}]
\textsc{List Coloring}, \textsc{Precoloring Extension}, \textsc{Boolean Constraint Satisfaction}, \textsc{Edge Disjoint Paths}, \textsc{Bayesian Network Structure Learning}, \textsc{Polytree Learning}, \textsc{Minimum Changeover Cost Arborescence}, and \textsc{Maximum Stable Roommates with Ties and Incomplete Lists} are fixed-parameter tractable w.r.t.\ slim tree-cut width.
\end{corollary}

Last but not least, given the ease with transferring dynamic programming algorithms from edge-cut width to slim tree-cut width, an inquisitive reader might be wondering whether it is not possible to formally prove that \emph{every} problem which is \FPT\ w.r.t.\ former is also \FPT\ w.r.t.\ the latter. That is, however, not true in general: one can construct entirely artificial problems which do not behave in this way. 

To illustrate this on a high level, let us consider an arbitrary graph problem $\PP$ which remains \NP-hard even on trees (as an example, the \textsc{Firefighter} problem~\cite{FinbowKMR07}) and can be solved on general $n$-vertex graphs in time $\tau(n)$. Moreover, let $\iota(n)$ denote the time required to compute the slim tree-cut width of a graph $G$ via an exhaustive brute force search, and let $\psi$ be a function which dominates both $\tau$ and $\iota$.
We now define an artificial new problem $\PP'$ as follows: 
\begin{itemize}
\item every $n$-vertex graph $G$ such that $\psi(\ecw(G))\leq n$
is a YES-instance, and otherwise
\item $G$ is a YES-instance if and only if $G$ is a YES-instance of \textsc{Firefighter}.
\end{itemize}

Then $\PP'$ is \FPT\ parameterized by edge-cut width. 
Indeed, given an instance $(G,k)$ of $\PP'$, one can attempt to run a brute-force search to determine the edge-cut width (which is promised to be at most $k$) with a time-out of $\psi(\psi(k))$. If the algorithm times out, this implies that $\psi(\ecw(G))\leq n$ and we correctly output ``Yes''. If not, we proceed by calling a brute-force algorithm to solve \textsc{Firefighter} on $G$, and this must once again complete in time at most $\psi(\psi(k))$. On the other hand, $\PP'$ remains \NP-hard even on graph classes with constant $\stcw(G)$---consider, for instance, the class of all graphs with two connected components, one of which ($C_1$) is a tree and the other ($C_2$) a graph from the class with constant slim tree-cut width but unbounded edge-cut width (one such class is depicted in Figure~2 of~\cite{ECW2022}). On some inputs from this class, $\PP'$ will ask for a solution to the \textsc{Firefighter} problem (which is \NP-hard on trees) but the parameter $\stcw(G)$ will remain constant.

\section{Conclusion}

The contribution of this work is mainly conceptual: it provides a possible resolution to the search for an alternative to treewidth for edge cuts which is both structurally sound and exhibits the expected (and desired) algorithmic properties. Slim tree-cut width can be viewed as the ``missing link'' which explains why the problems depicted in Table~\ref{tab:problems} admit fixed-parameter algorithms that exploit dynamic programming along small edge cuts w.r.t.\ both edge-cut width (as a generalization of the feedback edge number) and treewidth plus maximum degree. We firmly believe that there are many more problems of interest where edge-cut based parameters may help push the frontiers of tractability. On this front, the alternative characterization via the edge-cut width of a supergraph provides decompositions which are better suited for dynamic programming than tree-cut decompositions.

The problem of computing optimal decompositions for slim tree-cut width remains, similarly as in the case of tree-cut width~\cite{KimOPST18}, as a prominent open question. Moreover, we believe that the ideas used to obtain a 2-approximation algorithm for tree-cut width could also be used to obtain an improved constant-factor approximation for slim tree-cut width.

\bibliography{stcw.bib}
\end{document}